\documentclass[preview]{elsarticle}

\usepackage[english]{babel}
\usepackage[latin1]{inputenc}

\usepackage{amsmath}
\usepackage{amsthm}
\usepackage{latexsym}
\usepackage{graphicx}
\usepackage{amssymb}
\usepackage{hyperref}
\usepackage{float}

\usepackage[usenames,dvipsnames]{pstricks}
\usepackage{epsfig}

\usepackage{algorithm}
\usepackage[noend]{algpseudocode}

\newtheorem{theorem}{Theorem}

\newtheorem{prop}{Proposition}
\newtheorem{corollary}{Corollary}

\theoremstyle{remark}
\newtheorem{remark}{Remark}

\newcommand{\euler}{\phi}

\newcommand{\pell}{{\mathcal H}}

\newcommand{\para}{{\mathcal P}}

\newcommand{\ZZ}{{\mathbb Z}}
\newcommand{\KK}{{\mathbb K}}

\newcommand{\iso}{{\Phi}}
\DeclareMathOperator{\lcm}{lcm}

\journal{Journal of \LaTeX\ Templates}

\begin{document}

\begin{frontmatter}

\title{A multifactor RSA-like scheme with fast decryption 
based on R\'edei rational functions
over the Pell hyperbola}

\author{Emanuele Bellini}
\address{DarkMatter LLC, UAE}

\ead{eemanuele.bellini@gmail.com}



\author{Nadir Murru}
\address{University of Turin, Italy}

\ead{nadir.murru@unito.it}

\begin{abstract}
We propose a generalization of an RSA-like scheme 
based on R\'edei rational functions over the Pell hyperbola. 
Instead of a modulus which is a product of two primes,
we define the scheme on a multi-factor modulus, 
i.e. on a product of more than two primes.
This results in a scheme with a decryption which is quadratically faster, 
in the number of primes factoring the modulus, than the original RSA, 
while preserving a better security. 
The scheme reaches its best efficiency advantage over RSA 
for high security levels, 
since in these cases the modulus can contain more primes.
Compared to the analog schemes based on elliptic curves, 
as the KMOV cryptosystem, the proposed scheme is more efficient.
Furthermore a variation of the scheme with larger ciphertext size 
does not suffer of impossible group operation attacks, 
as it happens for schemes based on elliptic curves.
\end{abstract}


\end{frontmatter}


\section{Introduction}\label{sec:intro}
%
%
RSA is the most widespread asymmetric encryption scheme.
Its security is based on the fact that the trapdoor function 
$\tau_{N,e}(x) = x^e \mod N$, with $N=pq$ product of two large prime integers, and 
$e$ an invertible element in $\mathbb{Z}_{\euler(N)}$ ($\euler(N)$ being the Euler totient function), 
cannot be inverted by a polynomial-time in $\log N$ algorithm without knowing 
either the integers $p$, $q$, $\euler(N)$ or the inverse $d$ of $e$ modulo $\euler(N)$.
Thus the pair $(N,e)$, called the public key, is known to everyone, 
while the triple $(p,q,d)$, called the secret key, 
is only known to the receiver of an encrypted message.\\
Both encryption and decryption are performed through an exponentiation modulo $N$.
Precisely, the ciphertext $C$ is obtained as $C=M^e \pmod N$, and
the original message $M$ is obtained with the exponentiation $M = C^d \pmod N$.
While usually the encryption exponent is chosen to be small, 
the decryption exponent is about the size of $N$,
implying much slower performances during decryption with respect to encryption.\\
Through the years many proposal have been presented trying to speed up the decryption process.
In this work we present the fastest, to the authors knowledge,
of such decryption algorithms whose security is based on the factorization problem. \\
The presented scheme exploits different properties of R\'edei rational functions, 
which are classical functions in number theory.
The proposed decryption algorithm is quadratically, on the number of primes composing the modulus $N$, faster than RSA.

\vspace{5mm}

%
%
The work is divided as follows.
In Section \ref{sec:related_work} an overview of the main 
schemes based on the factorization problem which successfully improved RSA decryption step is presented.
In Section \ref{sec:product} the main theoretical results underlying our scheme are described. 
Section \ref{sec:critto} is devoted to the presentation of the cryptographic scheme, and in
Section \ref{sec:security} and \ref{sec:efficiency} 
its security and efficiency are discussed, respectively.
Section \ref{sec:conclusion} concludes the work.
\section{Related work}
\label{sec:related_work}
In this section we briefly overview
the main cryptographic schemes based on the factorization problem 
that have been introduced in order to improve RSA decryption step.\\
Usually, the general technique to speed up the RSA decryption step $C=M^e \pmod N$ 
is to compute the exponentiation modulo each factor of $N$ and 
then obtain $N$ using the Chinese Remainder Theorem.
\subsection{Multifactor RSA}
There exists variants of RSA scheme which exploit a modulus with more than 2 factors 
to achieve a faster decryption algorithm.
This variants are sometimes called Multifactor RSA (\cite{boneh2002fast}), 
or Multiprime RSA (\cite{ciet2002short}, \cite{compaq2000multiprime}).
The first proposal exploiting a modulus of the form $N=p_1 p_2 p_3$ 
has been patented by Compaq (\cite{compaq2000multiprime}, \cite{collins1998public}) in 1997.
About at the same time Takagi \cite{takagi1998fast} proposed an even faster solution using the modulus  $N=p^r q$,
for which the exponentiation modulo $p^r$ is computed using the Hensel lifting method \cite[p.137]{cohen2013course}.
Later, this solution has been generalized to the modulus $N=p^r q^s$ \cite{lim2000generalized}.\\
According to \cite{compaq2000multiprime} the appropriate number of primes to be chosen 
in order to resist state-of-the-art factorization algorithms 
depends from the modulus size, and, precisely, it can be:
up to 3 primes for 1024, 1536, 2048, 2560, 3072, and 3584 bit modulus, up to 4 for 4096, and up to 5 for 8192.
\subsection{RSA-like schemes}
Another solution which allows to obtain even faster decryption is to use RSA-like schemes based on isomorphism as 
\cite{koyama1992new}, \cite{koyama1995fast}, \cite{padhye2006public}, \cite{bellini2016}. 
As an additional property, these schemes owns better security properties with respect to RSA,
avoiding small exponent attacks to either $d$ (\cite{wiener1990cryptanalysis}) 
or $e$ (\cite{coppersmith1996low}, \cite{coppersmith1997small}),
and vulnerabilities which appear when switching from one-to-one communication scenario 
to broadcast scenario (e.g., see \cite{hastad1986n}).\\
The aforementioned schemes are based on isomorphism between two groups, 
one of which is the set of points over a curve, usually a cubic or a conic.
A complete overview on RSA-like schemes based on conics can be found in \cite{bellini2016}.
In general, schemes based on cubic curves have a computationally more expensive addition operation 
compared to schemes based on conic equations.
\subsection{Generalizing RSA-like scheme with multifactor modulus}
As done when generalizing from RSA to Multiprime RSA, 
in \cite{boudabra2017new} a generalization of \cite{koyama1992new}, \cite{koyama1995fast}
has been proposed, thus generalizing a RSA-like scheme based on elliptic curves and a modulus $N=pq$
to a similar scheme based on the generic modulus $N=p^rq^s$.\\
In this paper we present a similar generalization of the scheme \cite{bellini2016}, 
which is based on the Pell's equation, to the modulus $N=p_1^{e_1}\cdot\ldots\cdot p_r^{e_r}$ for $r > 2$, 
obtaining the fastest decryption of all schemes discussed in this section.
\section{Product of points over the Pell hyperbola}
\label{sec:product}
In \cite{bellini2016}, we introduced a novel RSA--like scheme 
based on an isomorphism between certain conics (whose the Pell hyperbola is a special case) and 
a set of parameters equipped with a non--standard product. 
In Section \ref{sec:critto}, we generalize this scheme considering 
a prime power modulus $N = p_1^{e_1} \cdots p_r^{e_r}$. 
In this section, we recall some definitions and properties given in \cite{bellini2016} 
in order to improve the readability of the paper. 
Then, we study properties of the involved products and sets in $\ZZ_{p^r}$ and $\ZZ_N$.
\subsection{A group structure over the Pell hyperbola over a field}
\label{sec:group}
Let $\mathbb K$ be a field and $x^2 - D$ an irreducible polynomial over $\mathbb K [x]$. Considering the quotient field $\mathbb A[x] = \mathbb K[x] / (x^2 - D)$, the induced product over $\mathbb A[x]$ is 

$$(p + q x) (r + s x) = (p r + q s D) + (q r + p s) x.$$

The group of unitary elements of 
$\mathbb A^*[x] = \mathbb A[x] - \{0_{\mathbb A[x]}\}$
\footnote{The element $0_{\mathbb A[x]}$ is the zero polynomial.}
is 
$\{ p + q x \in \mathbb A^*[x] : p^2 - D q^2 = 1 \}$. 
Thus, we can introduce the commutative group 
$(\pell_{D,\mathbb K}, \otimes)$, where 
$$\pell_{D,\mathbb K} = \{ (x,y) \in \mathbb K \times \mathbb K : x^2 - D y^2 =1 \}$$
and 
\begin{align}
(x,y) \otimes (w,z) = (x w + y z D, y w + x z), \quad \forall (x,y), \ (w,z) \in \pell_{D,\mathbb K}.
\end{align}\label{eq:pelloperation}
It is worth noting that $(1,0)$ is the identity and the inverse of an element $(x,y)$ is $(x,-y)$.

\begin{remark}
When $\mathbb K = \mathbb R$, the conic $\pell_{D,\mathbb K}$, for $D$ a non--square integer, is called the Pell hyperbola since it contains all the solutions of the Pell equation and $\otimes$ is the classical Brahamagupta product, see, e.g., \cite{jacobson2009}.
\end{remark}
\subsection{A parametrization of the Pell hyperbola}
\label{sec:param}
From now on let $\mathbb{A}=\mathbb{A}[x]$.\\
Starting from $\mathbb A^*$, we can derive a parametrization for $\pell_{D,\mathbb K}$. In particular, let us consider the group $\mathbb A^* / \mathbb K^*$, whose elements are the equivalence classes of $\mathbb A^*$ and can be written as
$$\{ [a + x] : a \in \mathbb K \} \cup \{[1_{\mathbb K^*}]\}.$$
The induced product over $\mathbb A^* / \mathbb K^*$ is given by 
$$[a + x] [b + x] = [ab + ax + bx + x^2] = [D + a b + (a + b) x]$$
and, if $a + b \not= 0$, we have
$$[a + x] [b + x] = \left[ \cfrac{D + a b}{a + b} + x \right]$$
else
$$[a + x] [b + x] = [D + a b] = [1_{\mathbb K^*}].$$
This construction allows us to define the set of parameters $\para_{\mathbb K} = \mathbb K \cup \{\alpha\}$, with $\alpha$ not in $\mathbb K$, equipped with the following product:
\begin{align}
\begin{cases} a \odot b = \cfrac{D + a b}{a + b}, \quad a + b \not= 0 \cr a \odot b = \alpha, \quad a + b = 0  
\end{cases}.
\end{align}\label{eq:prodpara}

We have that $(\para_\mathbb K, \odot)$ is a commutative group with identity $\alpha$ and the inverse of an element $a$ is the element $b$ such that $a + b = 0$. Now, consider the following parametrization for the conic $\pell_{D,\mathbb K}$:

$$y = \cfrac{1}{m}(x + 1)\,.$$
It can be proved that the following isomorphism between 
$(\pell_{D,\mathbb K}, \otimes)$ and $(\para_{\mathbb K}, \odot)$ holds:
\begin{equation} 
\iso_{D}:
\begin{cases} 
\pell_{D,\mathbb K} & \rightarrow \para_\KK \cr 
(x,y)  & \mapsto \cfrac{1+x}{y} \quad \forall (x,y)\in \pell_{D,\mathbb K}, \quad y\not=0 \cr 
(1,0)  & \mapsto\alpha \cr
(-1,0) & \mapsto 0 \ , 
\end{cases} 
\end{equation}
and
\begin{equation}\label{para} 
\iso^{-1}_{D}:
\begin{cases}
\para_\KK  & \rightarrow \pell_{D,\mathbb K} \cr 
m      & \mapsto \left(\cfrac{m^2+D}{m^2-D}\ , \cfrac{2m}{m^2-D}\right)\quad \forall m \in \mathbb K \cr 
\alpha & \mapsto(1,0)\ , 
\end{cases},
\end{equation}
see \cite{barbero2010} and \cite{bellini2016}.

\begin{prop}\label{prop:fermat}
When $\mathbb K = \ZZ_p$, $p$ prime, $(\para_{\mathbb K}, \odot)$ and $(\pell_{D,\mathbb K}, \otimes)$ are cyclic groups of order $p + 1$ and
$$m^{\odot (p+2)}=m \pmod p,\quad \forall m\in \para_{\ZZ_p}$$
or, equivalently
$$(x,y)^{\otimes (p+2)}=(x,y) \pmod p,\quad \forall (x,y)\in \pell_{D,\mathbb \ZZ_p},$$
where powers are performed using products $\odot$ and $\otimes$, respectively. See \cite{bellini2016}.
\end{prop}

The powers in $\para_{\mathbb K}$ can be efficiently computed by means of the R\'edei rational functions \cite{redei1946}, which are classical functions in number theory. They are defined by considering the development of
$$(z + \sqrt{D})^n = A_n(D,z) + B_n(D,z)\sqrt{D},$$
for $z$ integer and $D$ non--square positive integer. The polynomials $A_n(D,z)$ and $B_n(D,z)$ defined by the previous expansion are called R\'edei polynomials and can be evaluated by
$$M^n = \begin{pmatrix} A_n(D,z) & DB_n(D,z) \cr B_n(D,z) & A_n(D,z)  \end{pmatrix}$$
where
$$M = \begin{pmatrix} z & D \cr 1 & z \end{pmatrix}.$$
From this property, it follows that the R\'edei polynomials are linear recurrent sequences with characteristic polynomial $t^2 - 2 z t + (z^2 - D)$. The R\'edei rational functions are defined by 

$$Q_n(D,z) = \cfrac{A_n(D,z)}{B_n(D,z)}, \quad \forall n \geq 1.$$

\begin{prop}
Let $m^{\odot n}$ be the n--th power of $m \in \para_{\mathbb K}$ with respect to $\odot$, then 
$$m^{\odot n} = Q_n(D, m).$$
See \cite{barbero2010a}.
\end{prop}

\begin{remark}
The R\'edei rational functions can be evaluated by means of an algorithm of complexity $O(\log_2(n))$ with respect to addition, subtraction and multiplication over rings \cite{more1995}.
\end{remark}

\subsection{Properties of the Pell hyperbola over a ring}
\label{sec:zpr}

In this section, we study the case $\mathbb K = \ZZ_{p^r}$ that we will exploit in the next section for the construction of a cryptographic scheme. In what follows, we will omit from $\pell_{D, \mathbb K}$ the dependence on $D$ when it will be clear from the context.

First, we need to determine the order of $\pell_{\ZZ_{p^r}}$ in order to have a result similar to Proposition \ref{prop:fermat} also in this situation. 
\begin{theorem} \label{thm:pell-pr}
The order of the 
cyclic group $\pell_{\ZZ_{p^r}}$ 
is $p^{r-1} (p + 1)$, i.e., the Pell equation $x^2 - D y^2 = 1$ has $p^{r-1} (p + 1)$ solutions in $\ZZ_{p^r}$ for $D \in \ZZ_{p^r}^*$ quadratic non--residue in $\ZZ_{p}$.
\end{theorem}
\begin{proof}
Since, by Proposition \ref{prop:fermat}, 
the Pell equation in $\ZZ_p$ has $p + 1$ solutions, 
then we need to prove the following
\begin{enumerate}
 \item \label{lab:stepA} 
       any solution of the Pell equation in $\ZZ_p$,
       generates $p^{r-1}$ solutions of the same equation in $\ZZ_{p^r}$;
 \item \label{lab:stepB} 
       all the solutions of the Pell equation in $\ZZ_{p^r}$
       are generated as in the previous step.
\end{enumerate}
\begin{itemize}
 \item [(\ref{lab:stepA})]
 Let $(x_0, y_0)$ be a solution of $x^2 - D y^2 \equiv 1 \pmod p$.
 We want to prove that for any integer $0 \leq k < p^{r-1}$, 
 there exists one and only one integer $h$ such that $(x_0 + k p, y_0 + h p)$ is solution of $x^2 - D y^2 \equiv 1 \pmod{p^r}$.\\
 Indeed, we have
$$
(x_0 + k p)^2 - D (y_0 + h p)^2 
=
1 + v p + 2 x_0 k p + k^2 p^2 - 2 D y_0 h p - D h^2 p^2,
$$
since $x_0^2 - D y_0^2 = 1 + v p$ for a certain integer $v$. 
Thus, we have that $(x_0 + k p, y_0 + h p)$ is solution of $x^2 - D y^2 \equiv 1 \pmod {p^r}$ if and only if
$$D p h^2 + 2 D y_0 h - v - 2 x_0 k - k^2 p \equiv 0 \pmod{p^{r-1}}.$$
Hence, we have to prove that there is one and only one integer $h$ that satisfies the above identity. The above equation can be solved in $h$ by completing the square and reduced to
\begin{equation}\label{eq:compl-square} (2 D p h + 2 D y_0)^2 \equiv s \pmod{p^{r-1}}, \end{equation}
where $s = (2 D y_0)^2 + 4 (v + 2 x_0 k + k^2 p) D p$. 
Let us prove that $s$ is a quadratic residue in $\ZZ_{p^{r-1}}$. Indeed,
$$s = 4 D ((x_0 + k p)^2 - 1)$$
and surely the Jacobi symbol $\left( \cfrac{s}{p^{r-1}} \right) = \left( \cfrac{s}{p} \right)^{r-1} = 1$ if $r$ is odd. \\
If $r$ is even we have that
$$
\left( \cfrac{s}{p^{r-1}} \right) = 
\left( \cfrac{4}{p^{r-1}} \right) 
\left( \cfrac{D}{p^{r-1}} \right) 
\left( \cfrac{(x_0 + k p)^2 - 1}{p^{r-1}} \right)
=1
$$
since $\left( \cfrac{4}{p^{r-1}} \right) = 1$, 
$\left( \cfrac{D}{p^{r-1}} \right) = \left( \cfrac{D}{p} \right)^{r-1} = -1$ by hypothesis on $D$, 
$\left( \cfrac{(x_0 + k p)^2 - 1}{p^{r-1}} \right) = -1$, 
since $(x_0 + k p)^2 - 1 \equiv D y_0^2 \pmod p$. \\
Now, let $\pm t$ be the square roots of $s$. It is easy to note that 
$$t \equiv 2 D y_0 \pmod p, \quad -t \equiv - 2 D y_0 \pmod p$$
or 
$$-t \equiv 2 D y_0 \pmod p, \quad t \equiv - 2 D y_0 \pmod p.$$
Let us call $\bar t$ the only one between $t$ and $-t$ that is equal to $2 D y_0$ in $\ZZ_p$.
Hence, Equation \eqref{eq:compl-square} is equivalent to the linear equation 
$$p h \equiv (\bar t - 2 D y_0)(2 D)^{-1} \pmod{p^{r-1}},$$
which has one and only one solution, since $\bar t - 2 D y_0 \equiv 0 \pmod p$. Note that, if $\bar t$ is not equal to $2 D y_0$ in $\ZZ_p$ the above equation has no solutions.
Thus, we have proved that any solution of the Pell equation in $\ZZ_p$ generates $p^{r-1}$ solutions of the Pell equation in $\ZZ_{p^r}$.
 \item [(\ref{lab:stepB})]
 Now, we prove that all the solutions of the Pell equation in $\ZZ_{p^r}$ are generated as in step \ref{lab:stepA}.\\ 
Let $(\bar x, \bar y)$ be a solution of $x^2 - D y^2 \equiv 1 \pmod{p^r}$, i.e.,
$\bar x^2 - D \bar y^2 = 1 + w p^{r}$, for a certain integer $w$. 
Then $x_0 = \bar x - k p$ and $y_0 = \bar y -h p$, for $h, k$ integers, are solutions of $x^2 - D y^2 \equiv 1 \pmod p$. Indeed, 

$$
(\bar x - k p)^2 - D (\bar y - h p)^2 
= 
1 + w p^r - 2 \bar x k p + k^2 p^2 + 2 D \bar y h p - D h^2 p^2 \,.
$$
\end{itemize}
\end{proof}

As a consequence of the previous theorem, an analogous of the Euler theorem holds for the product $\otimes$.

\begin{theorem}
Let $p$, $q$ be prime numbers and $N = p^r q^s$, then for all $(x,y) \in \pell_{\ZZ_N}$ we have
$$(x,y)^{\otimes p^{r-1}(p + 1) q^{s-1}(s + 1)} \equiv (1, 0) \pmod N$$
for $D \in \ZZ_{N}^*$ quadratic non--residue in $\ZZ_p$ and $\ZZ_q$.
\end{theorem}
\begin{proof}
By Theorem \ref{thm:pell-pr}, we know that
$$(x,y)^{\otimes p^{r-1}(p + 1)} \equiv (1, 0) \pmod {p^r}$$
and
$$(x,y)^{\otimes q^{s-1}(s + 1)} \equiv (1, 0) \pmod {q^s}.$$
Thus, said $(a, b) = (x,y)^{\otimes p^{r-1}(p + 1) q^{s-1}(s + 1)}$, we have
$$(a, b) \equiv (1,0) \pmod {p^r},$$
i.e., $a = 1 + kp^r$ and $b = hp^r$ for some integers $h$, $k$. On the other hand, we have
$$(a, b) \equiv (1,0) \pmod {q^s} \Leftrightarrow (1 + kp^r, hp^r) \equiv (1,0) \pmod {q^s}.$$
We can observe that $1 + kp^r \equiv 1 \pmod {q^s}$ if and only if $k = k'q^s$ for a certain integer $k'$. Similarly, it must be $h = h'q^s$, for an integer $h'$. Hence, we have that
$(a, b) = (1 + k'p^rq^s, h'p^rq^s) \equiv (1, 0) \pmod N$.
\end{proof}

\begin{corollary}
Let $p_1, ..., p_r$ be primes and $N = p_1^{e_1}\cdot \ldots \cdot p_r^{e_r}$, then for all $(x,y) \in \pell_{\ZZ_N}$ we have
$$(x,y)^{\otimes \Psi(N)} = (1,0) \pmod N,$$
where
$$\Psi(N) = p_1^{e_1-1} (p_1 + 1) \cdot \ldots \cdot p_r^{e_r-1} (p_r + 1),$$
for $D \in \ZZ_{N}^*$ quadratic non--residue in $\ZZ_{p_i}$, for $i = 1, ..., r$.
\end{corollary}

Now, we can observe that when we work on $\ZZ_{N}$, the map $\iso_D$ is not an isomorphism. Indeed, the orders of $\pell_{D,\ZZ_{N}}$ and $\para_{\ZZ_{N}}$ do not coincide.
However, it is still a morphism and we also have 
$\lvert \ZZ_{N}^* \rvert = \lvert \pell^*_{\ZZ_{N}} \rvert$, 
because of the following proposition.
\begin{prop}
With the above notation, we have that
\begin{enumerate}
\item $\forall (x_1,y_1), (x_2,y_2)\in \pell^*_{\ZZ_{N}}$, $\iso_D(x_1,y_1)=\iso_D(x_2,y_2)\Leftrightarrow (x_1,y_1)=(x_2,y_2)$;
\item $\forall m_1, m_2\in \ZZ_{N}^*$, $\iso_D^{-1}(m_1)=\iso_D^{-1}(m_2)\Leftrightarrow m_1=m_2$;
\item $\forall m\in \ZZ_{N}^*$, we have $\iso^{-1}(m)\in \pell^*_{\ZZ_{N}}$ and $\forall (x,y)\in \pell^*_{\ZZ_{N}}$, we have $\iso_D(x,y)\in\ZZ_{N}^*$.
\end{enumerate}
See \cite{bellini2016}.
\end{prop}

As a consequence, we have an analogous of the Euler theorem also for the product $\odot$, i.e., for all $m \in \ZZ^*_N$ the following holds
$$m^{\odot \Psi(N)} = \alpha \pmod N\,,$$
where $\odot$ is the special product in $\para_{\ZZ_N}$ 
defined in Equation \ref{eq:prodpara}.
\section{The cryptographic scheme}\label{sec:critto}
In this section, we describe our public--key cryptosystem based on the properties studied in the previous section. 
\subsection{Key generation}
The key generation is performed by the following steps:
\begin{itemize}
\item choose $r$ prime numbers $p_1, \dots, p_r$, 
      $r$ odd integers $e_1, \dots, e_r$ and 
      compute $N = \prod_{i=1}^r p_i^{e_i}$;
\item choose an integer $e$ such that 
      $\gcd ( e, \lcm \prod_{i=1}^r{p_i^{e_i-1}(p_i + 1)} ) = 1$;
\item evaluate $d = e^{-1} \pmod{\lcm \prod_{i=1}^r{p_i^{e_i-1}(p_i + 1)}}$.
\end{itemize}
The public or encryption key is given by $(N, e)$ and the secret or decryption key is given by $(p_1,\ldots,p_r, d)$.
\subsection{Encryption}
We can encrypt pair of messages $(M_x, M_y) \in \ZZ_N^* \times \ZZ_N^*$, 
such that the following condition holds: $\left( \cfrac{M_x^2 - 1}{N} \right) = -1$. 
This will ensure that we can perform all the operations. 
The encryption of the messages is performed by the following steps:
\begin{itemize}
\item compute $D = \cfrac{M_x^2 - 1}{M_y^2} \pmod N$, so that $(M_x, M_y) \in \pell^*_{D, \ZZ_N}$;
\item compute $M = \iso(M_x, M_y) = \cfrac{M_x + 1}{M_y} \pmod N$;
\item compute the ciphertext $C = M^{\odot e} \pmod N = Q_e(D,M) \pmod N$
\end{itemize}
Notice that not only $C$, but 
the pair $(C,D)$ must be sent through the insecure channel.
\subsection{Decryption}
The decryption is performed by the following steps:
\begin{itemize}
\item compute $C^{\odot d} \pmod N = Q_d(D,C) \pmod N = M$;
\item compute $\iso^{-1}(M) = \left( \cfrac{M^2 + D}{M^2 - D}, \cfrac{2 M}{M^2 - D} \right) \pmod N$ for retrieving the messages $(M_x, M_y)$.
\end{itemize}
\section{Security of the encryption scheme}
\label{sec:security}
The proposed scheme can be attacked by solving one of the following problems:
\begin{enumerate}
 \item factorizing the modulus $N=p_1^{e_1}\cdot \ldots \cdot p_r^{e_r}$; 
       \label{item:fact}
  \item computing $\Psi(N) = p_1^{e_1-1}(p_1+1)\cdot\ldots\cdot p_r^{e_r-1}(p_r+1)$,
        or finding the number of solutions of the equation $x^2-Dy^2 \equiv 1 \mod N$,
        i.e. the curve order, which divides $\Psi(N)$;
       \label{item:phi}
 \item computing Discrete Logarithm problem either in 
       $(\pell^*_{\ZZ_N},\otimes)$ or in $(\para^*_{\ZZ_N},\odot)$;
       \label{item:dlog}
 \item finding the unknown $d$ in the equation $ed \equiv 1 \mod \Psi(N)$;
       \label{item:equ}
 \item finding an impossible group operation in $\para_{\ZZ_N}$;
 \item computing $M_x,M_y$ from $D$.
\end{enumerate}
\subsection{Factorizing \textit{N} or computing the curve order}
It is well known that the problem of factorizing $N=p_1^{e_1}\cdot \ldots \cdot p_r^{e_r}$
is equivalent to that of computing the Euler totient function 
$\phi(N)=p_1^{e_1-1}(p_1-1)\cdot\ldots\cdot p_r^{e_r-1}(p_r-1)$, 
e.g. see \cite{miller1975riemann} or \cite[Section 10.4]{shoup2009computational}.\\
In our case we need to show the following
\begin{prop}
The problem of factorizing $N$ is equivalent to computing the value 
$\Psi(N)=p_1^{e_1-1}(p_1+1)\cdot\ldots\cdot p_r^{e_r-1}(p_r+1)$ 
or the order of the group 
$\para^*_{\ZZ_N}$ (or equivalently of $\pell^*_{\ZZ_N}$),
which is a divisor of $\Psi(N)$. 
\end{prop}
\begin{proof}
Clearly, knowing the factorization of $N$ yields $\Psi(N)$.\\
Conversely, suppose $N$ and $\Psi(N)$ are known.
A factorization of $N$ can be found by applying Algorithm \ref{alg:factors} recursively.
\end{proof}
\begin{remark}
Algorithm \ref{alg:factors} is an adaptation of the general algorithm 
in \cite[Section 10.4]{shoup2009computational}, 
used to factorize $N$ by only knowing $\euler(N)$ (Euler totient function) and $N$ itself. 
The main idea of the Algorithm  \ref{alg:factors} comes from the fact that
$x^{\odot \Psi(N)}=1 \pmod N$ for all $x \in \ZZ^*_N$, which is the analog of the Euler theorem in $\para_{\ZZ_N}$. 
Notice that, because of Step \ref{step:rand}, 
Algorithm \ref{alg:factors} is a probabilistic algorithm. 
Thus, to find a non-trivial factor, it might be necessary to run the algorithm more than once. 
We expect that a deeper analysis of the algorithm 
will lead to a similar probabilistic behaviour than the algorithm in \cite{shoup2009computational}, 
which returns a non-trivial factor with probability $1/2$.
\end{remark}
\begin{algorithm}
\caption{Find a factor of $N$ by knowing $N$ and $\Psi(N)$}\label{alg:factors}
\begin{algorithmic}[1]
\Function{Find factor($N$,$\Psi(N)$)}{}
\State $h =0$
\State $t = \Psi(N)$
\While {IsEven($t$)}
  \State h = h + 1
  \State t = t / 2
\EndWhile

\State $a = Random(N - 1)$ \label{step:rand}
\State $d = \gcd(a, N)$
\If {$d \ne 1$}
  \State \Return $d$
\EndIf

\State $b = a^{\odot t} \mod N$
\For {$j=0, \ldots, h-1$}
  \State $d = \gcd(b+1,N)$
  \If {$d \ne 1$ or $d \ne N$}
    \State \Return $d$
  \EndIf
  \State $b = b^2 \mod N$
\EndFor
\State \Return $0$
\EndFunction
\end{algorithmic}
\end{algorithm}
Since we proved that the problems \ref{item:fact} and \ref{item:phi} are equivalent, 
we can only focus on the factorization problem.\\
According to \cite{compaq2000multiprime}, 
state-of-the-art factorization methods as the Elliptic Curve Method \cite{lenstra1987factoring} or 
the Number Field Sieve \cite{lenstra1993number}, \cite{bernstein1993general}
are not effective if in the following practical cases
\begin{itemize}
 \item $|N|=1024, 1536, 2048, 2560, 3072, 3584$ and 
       $N=p_1^{e_1} p_2^{e_2} p_3^{e_3}$ with $e_1+e_2+e_3 \le 3$
       and $p_i,i=1,2,3$ greater than approximately the size of $\sqrt[3]{N}$.
 \item $|N|=4096$ and 
       $N=p_1^{e_1} p_2^{e_2} p_3^{e_3} p_4^{e_4}$ with $e_1+e_2+e_3+e_4 \le 4$
       and $p_i,i=1,\ldots,4$ greater than approximately the size of $\sqrt[4]{N}$.
 \item $|N|=8192$ and 
       $N=p_1^{e_1} p_2^{e_2} p_3^{e_3} p_4^{e_4} p_5^{e_5}$ with $e_1+e_2+e_3+e_4+e_5 \le 5$
       and $p_i,i=1,\ldots,5$ greater than approximately the size of $\sqrt[5]{N}$.
\end{itemize}
Notice that currently, the largest prime factor found by the Elliptic Curve Method 
is a 274 bit digit integer \cite{zimm17ecmwebsite}.
Note also that the Lattice Factoring Method (LFM) of 
Boneh, Durfee, and Howgrave-Graham \cite{boneh1999factoring} 
is designed to factor integers of the form $N=p^u q$ only for large $u$.

\subsection{Computing the Discrete Logarithm}
Solving the discrete logarithm problem in a conic curve 
can be reduced to the discrete logarithm problem in the underlying finite field \cite{menezes1992note}.
In our case the curve is defined over the ring $\mathbb{Z}_N$. 
Solving the DLP over $\mathbb{Z}_N$ without knowing the factorization of $N$ is
as hard as solving the DLP over a prime finite field of approximately the same size.
As for the factorization problem, the best known algorithm to solve DLP on a prime finite field
is the Number Field Sieve.
When the size of $N$ is greater than 1024 then the NFS can not be effective.
\subsection{Solving the private key equation}
In the case of RSA, small exponent attacks 
(\cite{wiener1990cryptanalysis}, \cite{coppersmith1996low}, \cite{coppersmith1997small}) 
can be performed to find the unknown $d$ in the equation $ed \equiv 1 \mod \Psi(N)$.
Generalization of these attacks can be performed on RSA variants 
where the modulus is of the form $N=p_1^{e_1}p_2^{e_2}$ \cite{lu2017cryptanalysis}.
It has already been argued in \cite{bellini2016} and \cite{koyama1995fast}
that this kind of attacks fails when the trapdoor function is not a
simple monomial power as in RSA, as it is in the proposed scheme.
\subsection{Finding an impossible group operation}
In the case of elliptic curves over $\ZZ_N$, 
as in the generalized KMOV cryptosystem \cite{boudabra2017new}, 
it could happen that an impossible addition between two curve points occurs, 
yielding the factorization of $N$. 
This is due to the fact that the addition formula
requires to perform an inversion in the underlying ring $\ZZ_N$.
However, as shown by the same authors of \cite{boudabra2017new}, 
the occurrence of an impossible addition is very unlikely 
for $N$ with few and large prime factors.\\
In our case an impossible group operation may occur if 
$a+b$ is not invertible in $\ZZ_N$, i.e. if $\gcd(a+b,N) \ne 1$,
yielding in fact a factor of $N$.
However, also in our case, if $N$ contains a few large prime factors, 
impossible group operations occur with negligible probability, 
as shown by the following proposition.
\begin{prop}
The probability to find an invertible element in $\para_{\ZZ_N}$ is approximately
$$
1-\left( 1-\frac{1}{p_1}\right)\cdot\ldots\cdot\left(1-\frac{1}{p_r}\right)
$$
\end{prop}
\begin{proof}
The probability to find an invertible element in $\para_{\ZZ_N}$ is given by 
dividing the number of non-invertible elements in $\para_{\ZZ_N}$ by the total number of elements of this set, as follows:
\begin{align}
  & \frac{|\para_{\ZZ_N}| - \#\{\text{invertible elements in }\para_{\ZZ_N} \}}
    {|\para_{\ZZ_N}|} 
= \\
= & \frac{|\ZZ_N|+1 - (\#\{\text{invertible elements in }\ZZ_N\}+1)}
    {|\ZZ_N|+1} 
= \\
= & \frac{N-\phi(N)}{N+1} 
= \\
\sim & 1-\left( 1-\frac{1}{p_1}\right)\cdot\ldots\cdot\left(1-\frac{1}{p_r}\right)
\end{align}
where we used $N\sim N+1$ and 
$\phi(N)=N\left( 1-\frac{1}{p_1}\right)\cdot\ldots\cdot\left(1-\frac{1}{p_r}\right)$.
\end{proof}
This probability tends to zero for large prime factors.

\vspace{5mm}

Let us notice that, in the Pell curve case, 
it is possible to avoid such situation,
by performing encryption and decryption in $\pell^*_{\ZZ_N}$, 
without exploiting the isomorphism operation. 
Here the group operation $\otimes$ 
is defined between two points on the Pell curve, 
as in Equation \ref{eq:pelloperation},
and does not contain the inverse operation.
In the resulting scheme the ciphertext is obtained as
$(C_x,C_y)=(M_x,M_y)^{\otimes e}$, where the operation $\otimes$ depends on $D$. 
Thus the triple $(C_x,C_y,D)$ must be transmitted, resulting in a 
non-compressed ciphertext.
\subsection{Recovering the message from \textit{D}}
To recover the message pair $(M_x, M_y)$ from 
$D = \frac{M_x^2 - 1}{M_y^2} \pmod N$, 
the attacker must solve the quadratic congruence
$M_x^2 - D M_y^2 - 1 = 0 \pmod N$
with respect to the two unknowns $M_x$ and $M_y$.
Even if one of the two coordinates is known (partially known plaintext attack),
it is well known that computing square roots modulo a composite integer $N$,
when the square root exists, is equivalent to factoring $N$ itself.
\subsection{Further comments}
As a conclusion to this section, we only mention that as shown in \cite{bellini2016},
RSA-like schemes based on isomorphism own the following properties:
they are more secure than RSA in the broadcast scenario,
they can be transformed to semantically secure schemes using standard techniques
which introduce randomness in the process of generating the ciphertext.
\section{Efficiency of the encryption scheme}
\label{sec:efficiency}
Recall that our scheme encrypts and decrypts messages of size $2\log N$.
To decrypt a ciphertext of size $2\log N$ using CRT, 
standard RSA requires four full exponentiation modulo $N/2$-bit primes. 
Basic algorithms to compute $x^d \mod p$ requires $O(\log d \log^2p)$, 
which is equal to $O(\log^3p)$ if $d \sim p$.\\
Using CRT, if $N=p_1^{e_1}\cdot \ldots \cdot p_r^{e_r}$, 
our scheme requires at most $r$ exponentiation modulo $N/r$-bit primes.\\
This means that the final speed up of our scheme with respect to RSA is
\begin{align}
 \frac{4 \cdot (N/2)^3}{r \cdot (N/r)^3} = r^2/2
\end{align}
When $r=2$ our scheme is two times faster than RSA, 
as it has already been shown in \cite{bellini2016}.
If $r=3$ our scheme is $4.5$ time faster, 
with $r=4$ is 8 times faster,
and with $r=5$ is $12.5$ times faster.
\section{Conclusions}
\label{sec:conclusion}
We generalized an RSA-like scheme based on the Pell hyperbola 
from a modulus that was a product of two primes to a generic modulus.
We showed that this generalization leads to a very fast decryption step, 
up to 12 times faster than original RSA 
for the security level of a modulus of 8192 bits. 
The scheme preserves all security properties of RSA-like schemes, 
which are in general more secure than RSA, especially in a broadcast scenario. 
Compared to similar schemes based on elliptic curves 
it is more efficient.
We also pointed that a variation of the scheme with non-compressed ciphertext 
does not suffer of impossible group operation attacks.
%
%
\section*{References}
\bibliography{biblio-critto}
\end{document}